\pdfoutput=1
\documentclass[journal]{IEEEtran}

\usepackage[T1]{fontenc}
\usepackage{amsmath,amssymb,amsthm}
\usepackage{booktabs}
\usepackage{cite}
\usepackage{microtype}
\usepackage{tikz}
\usepackage{pgfplots}
\pgfplotsset{compat=1.17}
\usetikzlibrary{arrows.meta,decorations.pathreplacing,positioning}
\usepackage[colorlinks=true,linkcolor=black,citecolor=black,urlcolor=blue]{hyperref}

\newtheorem{theorem}{Theorem}
\newtheorem{corollary}{Corollary}
\newtheorem{proposition}{Proposition}
\theoremstyle{remark}
\newtheorem{remark}{Remark}

\newcommand{\Ab}{\bar{A}}
\newcommand{\Aone}{A_{1,1}}
\newcommand{\Atwo}{A_{2,1}}
\newcommand{\Rel}{\operatorname{Re}}
\newcommand{\Iml}{\operatorname{Im}}

\begin{document}

\title{Time Invariance, Circle Symmetry, and the\\ Completeness of the Cardiff Behavioral Model}

\author{Nicholas~B.~Tufillaro%
\thanks{N.~B.~Tufillaro is with Aqualytics, Corvallis, OR, USA (e-mail:
nick@aqualytics.eco).}%
\thanks{Preprint, 22 September 2026; revised 27 November 2026.}}

\markboth{Preprint}{Tufillaro: Time invariance, circle symmetry, and the Cardiff model}

\maketitle

\begin{abstract}
Frequency-domain behavioral models of nonlinear microwave devices (the Cardiff model,
X-parameters, and higher-order sinusoidal describing functions) share one functional
form, and each modeling framework justifies the form by its own argument, time invariance
among them. For a two-port under single-tone excitation with only the fundamental
incident, we collect in one theorem, with a proof, the argument by which time invariance
alone determines the model form, and note its extension to any number of ports. Under a
shift of the time origin each harmonic phasor rotates by a multiple of the fundamental's
angle, so the spectral map of a time-invariant device is equivariant under a weighted
circle action, and every such response is a phase factor attached to the fundamental
multiplying a function of the wave magnitudes and their relative phase. The theorem
restates known results and gives no truncation error. In the monomial rewriting of Tian
\emph{et al.} the Cardiff exponents are counters: $m$ is the order of the load-side
nonlinearity, $n$ is set by the drive-side harmonic, and $r$ is the number of conjugate
pairs. The elementary bound $r_{\max}=\lfloor K/2\rfloor$ for a load-side nonlinearity of
degree $K$ makes the familiar restriction $r\le1$ a hypothesis about the device, exact
for a cubic nonlinearity; for a loaded device the bound is replaced by a decay in $r$,
which in an example solved exactly is governed by the ratio of the measured radius to the
radius of convergence of the series.
The statements are illustrated with a published measurement and two simulations; they
have not been tested against a measurement in which $K$ is known independently.
\end{abstract}


\section{Introduction}

\IEEEPARstart{F}{requency-domain} behavioral models describe a nonlinear device by the
map from the incident-wave phasors at its ports to the scattered-wave phasors on a
harmonic grid. The Cardiff model~\cite{qi2009,woodington2010,azad2022}, the X-parameter
framework~\cite{root2005phd,verspecht2006phd,root2013} and the higher-order sinusoidal
input describing functions (HOSIDF) of the mechanical-systems
literature~\cite{nuij2006} all write this map in the same way: a phase factor
attached to the fundamental, multiplying an expansion in the magnitudes of the incident
waves and their relative phases. Each was justified on physical grounds by its authors
--- Woodington \emph{et al.}
observe that the device ``responds only to the magnitudes of the signals \ldots\ and the
relative phase difference''~\cite{woodington2010}, and Tasker and Benedikt write the general two-port
response at harmonic $h$ as the $h$th power of the drive phase times an arbitrary function of the two magnitudes and the
relative phase, Fourier-expanded in the latter~\cite{tasker2011,tasker2009dml}; Root
\emph{et al.} call the phase factor and the magnitude-only dependence of the
coefficients ``a necessary consequence of the assumed time invariance of the underlying
system''~\cite{root2005phd}; Nuij
\emph{et al.} introduce a ``virtual harmonics generator''~\cite{nuij2006}. In the baseband modeling of power amplifiers for digital
predistortion, Enzinger proves, for a general time-invariant operator with memory, that
the $k$-th harmonic baseband operator obeys $N_k(e^{j\xi}x)=e^{jk\xi}N_k(x)$, names the
property ``$k$-th order phase homogeneity,'' and calls it ``a necessary symmetry of all
complex baseband models of time-invariant passband systems''~\cite{enzinger2018}. The
odd-order structure of the baseband Volterra series, in which a term of degree $p$ has
$(p+1)/2$ unconjugated and $(p-1)/2$ conjugated factors, is that symmetry's consequence
for polynomial models at the fundamental, and is older still~\cite{benedetto1979,
zhou2005}. In this paper we recast these observations as an invariant, geometric account of
the origin of the model form and of its exponents, with one weight per harmonic and one
variable per port. The question addressed is narrow: what does time invariance alone
require of the map from incident to scattered phasors, and what do the three Cardiff
exponents count? Three points are made.

\emph{The form is complete.} Time invariance is the statement that the device response
is unchanged by a shift of the time origin. On the harmonic grid that shift is a
rotation of every phasor by an angle proportional to its harmonic index, i.e.\ an
action of the circle group $U(1)$ with integer weights, and the spectral map is
\emph{equivariant} under it. The invariants and equivariants of a compact group acting
linearly are classical objects, and their structure --- finitely many generators, with
Schwarz's theorem~\cite{schwarz1975} extending the description from polynomials to
smooth functions --- gives the general solution (Theorem~\ref{thm:complete}). For a
two-port in single-tone periodic steady state the Cardiff form, as Tasker and Benedikt wrote
it~\cite{tasker2011,tasker2009dml}, is that solution. The index relation $m=|n|+2r$ of the Cardiff model is the condition that the response
be smooth at zero load wave. The predistortion literature
draws the same distinction for one input, between the odd-order baseband monomials
$x^{k+1}\bar x^{k}$ and the even-order terms $|x|^{2k+1}x$, which satisfy the symmetry but
are not polynomial in $x$ and $\bar x$~\cite{enzinger2018,enzinger2016,ding2004}; the
index relation is its two-variable form.

\emph{The exponents acquire physical interpretations.} Tian \emph{et al.} have rewritten
the Cardiff model as a polynomial in the load-side wave $A$ and its conjugate, with the
term $(A)^{m-k}(\Ab)^{k}$ and $n=m-2k$, and noted that $m$ is then the polynomial
order~\cite{tian2025}. In that rewriting, with $a=m-k$ and $b=k$, $m=a+b$, $n=a-b$ and
$r=\min(a,b)$: the magnitude exponent is the order of the load-side nonlinearity that
produced the term, the phase exponent is set by the drive-side harmonic, and $r$ counts
conjugate pairs (Section~\ref{sec:exponents}). From this,
$r_{\max}=\lfloor K/2\rfloor$ for a device whose dependence on the load-side voltage
is a polynomial of degree $K$ (Theorem~\ref{thm:rmax}), and the ``mixing order'' by which
the Cardiff literature chooses its exponents~\cite{tian2025,alrawachy2025} is recovered
as $m+|n-h|$, the least total nonlinearity order able to produce the term.

\emph{The restriction $r\le1$ can be interpreted as a property of the device.} It is exact when the
load-side nonlinearity is at most cubic and fails generically at fourth order, where
the first new term is $|A|^{4}$. For a loaded device, whose load-side voltage depends on
its own output through the load, every order is present and the restriction becomes a truncation governed by a decay of the
coefficients in $r$; the decay rate is measurable from the fit and gives a stopping
rule. In an example whose series is known exactly the step ratio approaches
$(\rho/R)^{2}$, with $\rho$ the radius of the measured region and $R$ the radius of
convergence, and a standard estimate gives a truncation bound of that form
(Section~\ref{sec:loaded}). The tailored A-pull identification of Tasker and
co-workers~\cite{alrawachy2025,tasker2020}, which modulates the load-side wave and reads
the exponent pairs off a spectrum, measures the decomposition directly, and its
published data on a GaN device~\cite{alrawachy2025} --- nine exponent pairs, all with
$r\le1$, none with $m\ge4$ --- is consistent with an effective load-side order of three
(Section~\ref{sec:apull}). A polynomial device reproduces the pattern in simulation. The
ASM-HEMT compact model of a GaN high-electron-mobility transistor
(HEMT)~\cite{khandelwal2018asmhemt} in ngspice, with its default parameters, has a weaker
load-side nonlinearity (largest detected $m=2$ at the same floors) and shows the cluster
structure without reproducing the measured set of pairs (Section~\ref{sec:sims}).

The invariant and equivariant
structure of weighted circle and torus actions is textbook material; the standard
reference is Golubitsky, Stewart and Schaeffer~\cite[Ch.~XII and~XVI]{gss1988}, and the
finiteness theorem is in Weyl's own
lectures~\cite{weyl1936}. Coordinate-free descriptions of nonlinear input--output maps
exist --- Chen--Fliess series~\cite{brockett1976}, port-Hamiltonian
interconnection~\cite{vanderschaft2014}, the multi-time envelope formulation that
places the \emph{signal} on a torus~\cite{roychowdhury2001} --- and the linearization
of the spectral map about a large-signal drive is a harmonic transfer matrix in the
sense of the linear time-periodic literature~\cite{wereley1990,sandberg2005}. Most of what follows is restated. The restatements are the model form
(Tasker and Benedikt~\cite{tasker2011,tasker2009dml}; Root \emph{et al.}~\cite{root2005phd}); its
derivation from a time shift (Enzinger~\cite{enzinger2018} for one input, and the phase
normalization of~\cite{root2005phd,root2013}); the behavior at zero load wave, which is
the two-variable version of the even-order baseband
result~\cite{enzinger2018,enzinger2016,ding2004}; and the monomial rewriting, with the
identification of $m$ as the polynomial order (Tian \emph{et al.}~\cite{tian2025}).
Theorem~\ref{thm:complete} collects these in one statement with a proof. What is added is
the identification $r=\min(a,b)$ with the elementary bound that follows from it
(Theorem~\ref{thm:rmax}), the observation that for a loaded device the bound is replaced
by a decay in $r$ (Section~\ref{sec:loaded}), and the interpretation of the tailored
A-pull cluster width as the load-side order (Section~\ref{sec:apull}). None of the added
statements has been tested against a measurement in which the degree $K$ is known
independently. The correspondence relates the Cardiff, X-parameter and
describing-function forms for single-tone excitation, and the single-input
phase-homogeneity result of the predistortion literature~\cite{enzinger2018} is the same
symmetry; the X-parameter and Cardiff literatures cite each other, and neither cites the
describing-function or predistortion work.
Earlier work on behavioral modeling from the dynamical-systems
side~\cite{root2003dac,wood2004,wood2005chapter} concerned memory and time-domain
embedding; the present paper is confined to single-tone periodic steady state, where memory
does not enter. The companion paper~\cite{tufillaro2026memory} treats memory under
modulation; it uses the equivariance of Section~\ref{sec:symmetry} and shares no other
result with this one. The tutorial~\cite{tufillaro2026tutorial}
presents Sections~\ref{sec:symmetry} to~\ref{sec:exponents} at an elementary level with
worked examples and has no result that is not here.

\section{Time Invariance as a Circle Action}
\label{sec:symmetry}

Consider a two-port in periodic steady state at fundamental $\omega_0$. The incident
wave at port $p$ is $a_p(t)=\Rel\sum_{k\ge0}A_{p,k}e^{jk\omega_0 t}$ and the scattered
wave $b_p(t)=\Rel\sum_{h\ge0}B_{p,h}e^{jh\omega_0t}$; waves are power-normalized so that
$|A_{p,k}|^{2}$ is a power. A time-invariant device in periodic steady state defines a map
$B=F(A)$ from the incident phasors to the scattered ones.

Start the clock at time $\tau$, and consider a later time $t=t'+\tau$, which the shifted
clock reads as $t'$. The waveforms, the device and the experiment are unchanged, but the
phasors that describe the same waveforms in terms of $t'$ are multiplied by phase factors
that depend on the harmonic index:
\begin{equation}
  A_{p,k}\longmapsto A_{p,k}e^{jk\theta},\qquad B_{p,h}\longmapsto B_{p,h}e^{jh\theta},
  \qquad \theta=\omega_0\tau .
  \label{eq:action}
\end{equation}
The $k$-th harmonic rotates $k$ times as fast as the fundamental, and a shift by one
period leaves every phasor unchanged, so $\theta$ is an angle on the circle. The transformations
\eqref{eq:action} compose by adding angles: they form the group
$U(1)=\{e^{j\theta}\}$, the simplest compact Lie group, acting on the phasors with
integer \emph{weights} $k$ and $h$ (Fig.~\ref{fig:weights}); the phasor space thus
carries a representation of $U(1)$ that is a direct sum of one-dimensional
representations, the characters $e^{jk\theta}$, and the harmonic index is the label of
the irreducible representation. The weight is a winding number: as $\theta$ goes once
around the circle, a quantity of weight $k$ goes $k$ times around the unit circle, in the
negative sense if $k<0$, which is why conjugation reverses the sign of the weight. Time
invariance is the
statement that $F$ commutes with this action,
\begin{equation}
  F\bigl(\{A_{p,k}e^{jk\theta}\}\bigr)_{p,h} = e^{jh\theta}F(\{A_{p,k}\})_{p,h},
  \label{eq:equiv}
\end{equation}
i.e.\ $F$ is \emph{equivariant}: the output of weight $h$ rotates at weight $h$ when
the inputs rotate at their weights. Equation~\eqref{eq:equiv} is the whole physical
input to Section~\ref{sec:complete}.

\begin{figure}[!t]
\centering
\begin{tikzpicture}[scale=0.72]
  \foreach \k/\x in {1/0, 2/3.6, 3/7.2}{
    \begin{scope}[shift={(\x,0)}]
      \draw[black!25] (-1.5,0)--(1.5,0);
      \draw[black!25] (0,-1.5)--(0,1.5);
      \draw[black!55] (0,0) circle (1.2);
      \draw[-{Stealth[length=2.2mm]},thick] (0,0)--(25:1.2);
      \node at (25:1.5) {\footnotesize $A_{p,\k}$};
      \pgfmathsetmacro\ang{25+\k*38}
      \draw[-{Stealth[length=2.2mm]},thick,black!55,dashed] (0,0)--(\ang:1.2);
      \draw[-{Stealth[length=1.8mm]},black!70] (25:0.72) arc (25:\ang:0.72);
      \node at (0,-1.85) {\footnotesize $k=\k$: rotates by $\k\theta$};
    \end{scope}}
\end{tikzpicture}
\caption{Starting the clock at time $\tau$ rotates the $k$-th harmonic phasor by
$k\theta$, $\theta=\omega_0\tau$: all harmonics rotate together, at rates set by their
harmonic index, which is their weight under the circle action.}
\label{fig:weights}
\end{figure}
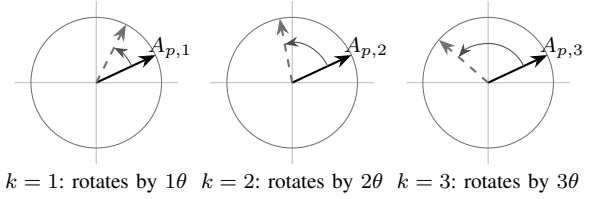

For a single tone and its harmonics the time origin is the only such freedom. When the
device is driven by several incommensurate tones each tone carries its own phase origin:
two incommensurate tones are represented by a two-torus, $T^{2}=U(1)\times U(1)$, one
circle per tone, and several by a torus with one circle per tone. Commensurate tones share
a fundamental and reduce to a single circle with the appropriate integer weights. Moving a
reference plane is a freedom of a different kind. It rotates the waves at that
port~\cite{kurokawa1965}, and so replaces the map $F$ by a different map, as it changes
$S$-parameters; it places no constraint on $F$. Only the single circle is used in this
paper.

\section{Invariants, Equivariants and Completeness}
\label{sec:complete}

\subsection{Invariants of the circle action}

A function of the phasors is \emph{invariant} if it is unchanged by every element of the
group. For the two fundamental incident waves $\Aone$ (drive side) and $\Atwo$ (load
side), both of weight one, two kinds of quantity survive (i.e.\ are invariant): the magnitudes $|\Aone|$,
$|\Atwo|$, and the weight-balanced products such as $\Atwo\Ab_{1,1}$, whose phase factors
cancel. In general the monomial $\prod A_{p,k}^{\alpha_{p,k}}\Ab_{p,k}^{\beta_{p,k}}$ is
invariant exactly when its total weight $\sum k(\alpha_{p,k}-\beta_{p,k})$ vanishes.
Dividing the balanced product
by the magnitudes leaves the unit phasor of the relative phase,
\begin{equation}
  Q\equiv\frac{\Atwo/|\Atwo|}{\Aone/|\Aone|}=e^{j(\phi_{2,1}-\phi_{1,1})},
  \qquad \phi_{p,k}=\angle A_{p,k},
\end{equation}
which is the object the Cardiff literature writes as $\angle\Atwo/\angle\Aone$. (In that
literature $\angle A$ denotes the unit phasor $e^{j\angle A}$, not the angle; the
notation is kept here.)

Two classical results describe the invariants. For a compact group acting linearly, the
ring of invariant polynomials is finitely generated (Hilbert; Weyl~\cite{weyl1936}): there
is a finite set of invariant polynomials, a Hilbert basis, such that every invariant
polynomial is a polynomial in the basis elements. For the two fundamental incident waves
the basis consists of $|\Aone|^{2}$, $|\Atwo|^{2}$, $\Rel(\Atwo\Ab_{1,1})$ and
$\Iml(\Atwo\Ab_{1,1})$, subject to the single relation
$\Rel^{2}+\Iml^{2}=|\Aone|^{2}|\Atwo|^{2}$. Schwarz's theorem~\cite{schwarz1975} extends
the result from polynomials to smooth functions: every smooth invariant is a smooth
function of the same generators. Neither theorem is needed where both waves are nonzero.
There the magnitudes and the relative phase are smooth coordinates on the quotient, and
the division by the drive phase that fixes the phase reference in the X-parameter
literature~\cite{root2005phd,root2013} is itself the proof. An argument is needed only at
$\Atwo=0$, where those coordinates are singular, and the second half of
Theorem~\ref{thm:complete} uses Schwarz's theorem there.

\subsection{Equivariants and the completeness theorem}

The response $B_{p,h}$ is not invariant; it has weight $h$. Equivariants are described
by a standard construction. Pick any weight-$h$ quantity, a \emph{carrier}, and divide the
output by the carrier; the quotient has weight zero and is therefore an invariant, hence
a smooth function of the Hilbert basis. The natural carrier is the drive-side phase,
$P\equiv e^{j\phi_{1,1}}=\Aone/|\Aone|$, of weight one. Multiplying the invariant
quotient back by the carrier gives the general equivariant.

\begin{theorem}[Completeness]\label{thm:complete}
Let a two-port in single-tone periodic steady state be time invariant, and let the
scattered-wave phasors depend smoothly on the fundamental incident phasors $\Aone,\Atwo$
away from $\Aone=0$ and $\Atwo=0$. Then
\begin{equation}
  B_{p,h}=P^{h}\sum_{n\in\mathbb Z} f_{p,h,n}\bigl(|\Aone|,|\Atwo|\bigr)\,Q^{n},
  \label{eq:complete}
\end{equation}
with the $f_{p,h,n}$ smooth functions of two real variables, and every map of this form
is time invariant. If in addition the response is smooth at $\Atwo=0$, the coefficient
of $Q^{n}$ is $|\Atwo|^{|n|}$ times a smooth function of $|\Aone|$ and $|\Atwo|^{2}$;
its Taylor expansion in $|\Atwo|^{2}$ gives the Cardiff form
\begin{equation}
  B_{p,h}=(\angle\Aone)^{h}\sum_{r\ge0}\sum_{n}C_{p,h,m,n}|\Atwo|^{m}
  \Bigl(\tfrac{\angle\Atwo}{\angle\Aone}\Bigr)^{n},
  \label{eq:cardiff}
\end{equation}
with $m=|n|+2r$, an equality where the response is analytic in $\Atwo$ and its conjugate
and an asymptotic expansion with remainder otherwise.
\end{theorem}

\begin{proof}
On the open set $\Aone\ne0$, $\Atwo\ne0$ the two magnitudes and the relative phase
$\Delta\phi=\phi_{2,1}-\phi_{1,1}$ are smooth coordinates on the quotient by the circle,
so the weight-zero function $B_{p,h}P^{-h}$ is a smooth function of
$(|\Aone|,|\Atwo|,\Delta\phi)$, $2\pi$-periodic in $\Delta\phi$; its Fourier series in
$\Delta\phi$ is \eqref{eq:complete}. Conversely each term of \eqref{eq:complete} has
weight $h$. For the second statement fix $\Aone$ and $n\ge0$, and let $g(z)$ be the
coefficient of $Q^{n}$ as a function of $z=\Atwo$. It is the weight-$n$ component of a
smooth function under rotation of $z$ alone, so it is smooth at $z=0$ and its Taylor
series there contains only monomials $z^{a}\bar z^{b}$ with $a-b=n$. The product
$\bar z^{\,n}g(z)$ is then a smooth invariant of the single variable $z$, hence a smooth
function of $|z|^{2}$ by Schwarz's theorem~\cite{schwarz1975}, and it vanishes to order
$n$ in $|z|^{2}$. Therefore $g(z)=z^{n}G(|z|^{2})$ with $G$ smooth; the case $n<0$ follows
by conjugation.
\end{proof}

\begin{remark}[Number of ports]\label{rem:ports}
Nothing in the proof is specific to two ports. For a device with $N$ ports driven at
the fundamental alone, the incident phasors $A_{1,1},\dots,A_{N,1}$ all have weight
one, and the invariant polynomials are generated by the products $A_{q,1}\Ab_{p,1}$,
$1\le p\le q\le N$, that is, by the $N$ squared magnitudes and the real and imaginary
parts of $N(N-1)/2$ products, $N^{2}$ real generators. Where $A_{1,1}\ne0$ the $N$
magnitudes and the $N-1$ phases relative to port~1 are coordinates on the quotient, and
the same construction, the carrier $P^{h}$ times an invariant, gives
\begin{equation}
  B_{p,h}=P^{h}\sum_{\mathbf n\in\mathbb Z^{N-1}} f_{p,h,\mathbf n}\bigl(|A_{1,1}|,\dots,|A_{N,1}|\bigr)
  \prod_{q=2}^{N} Q_q^{\,n_q},
  \label{eq:multiport}
\end{equation}
with $\mathbf n=(n_2,\dots,n_N)$ and $Q_q=\angle A_{q,1}/\angle A_{1,1}$, a Fourier
series in the $N-1$ relative phases; smoothness at $A_{q,1}=0$ imposes $m_q=|n_q|+2r_q$ on the power of $|A_{q,1}|$
multiplying $Q_q^{n_q}$, port by port, by the argument of the second statement. An
incident harmonic $A_{q,k}$ with $k>1$ enters with weight $k$; the multi-harmonic form is
written out in~\cite{tasker2011,bell2014} and is not treated here.
\end{remark}

We want to emphasize three points about Theorem~\ref{thm:complete}. First, the representation\footnote{Here ``representation'' is used in the sense of a
representation formula, a way of writing every element of a function space, not in the
group-theoretic sense of Section~\ref{sec:symmetry}; \eqref{eq:complete} describes the
equivariant maps between two representations of $U(1)$.} \eqref{eq:complete} is not an
approximation, a truncation or a
modeling assumption; it is the general solution, and the index $n$ is canonical rather
than chosen: the sum over $n$ is the Fourier series in the relative phase $\Delta\phi$, so
$n$ is a weight, a property of the group, not of the modeler. It is the form Tasker and
Benedikt give for fundamental load-pull, $b_{n,h}=P_1^{h}\,g(|a_{1,1}|,|a_{2,1}|,Q_1/P_1)$
at fixed bias, with $g$ expanded in powers of the periodic phase vector $Q_1/P_1$ and
coefficients that are functions of the two magnitudes~\cite{tasker2011}; it appears
earlier in Tasker's 2009 lecture~\cite{tasker2009dml}. It is also the form that Woodington \emph{et al.}~\cite{woodington2010} fitted; what is
collected here is the statement, with proof, that the form is exhaustive for smooth responses. The derivation of
the X-parameter form in~\cite{root2013}, which fixes the phase reference by dividing every
phasor by the drive phase, is already such a proof wherever the drive wave is nonzero; the present statement writes out the
Fourier series in the relative phase and the behavior at $\Atwo=0$. The theorem gives no
bound on the error of a truncated fit and says nothing about the conditioning of the
least-squares problem. The coefficients $C_{p,h,m,n}$ are written $K_{p,h,m,n}$ in the
Cardiff literature; $K$ is kept here for the degree of the load-side nonlinearity. The second statement
is, for one input, the distinction Enzinger draws
between the odd-order baseband monomials $x^{k+1}\bar x^{k}$ and the even-order terms
$|x|^{p-1}x$ with $p$ even, which satisfy phase homogeneity but are not polynomial in $x$
and $\bar x$~\cite{enzinger2018}. Second, memorylessness
was not assumed. A device with thermal or trapping memory held in continuous-wave (CW) steady state obeys
\eqref{eq:complete}, because the only physical hypothesis was that the response is
unchanged when the clock is re-zeroed; memory alters the response under modulation, which no CW measurement detects.
This is the conclusion the X-parameter literature also reaches~\cite{root2013}, and the
dynamic gain model of Verspecht \emph{et al.}~\cite{verspecht2023dg} imposes the same
structure at the envelope level: the response is the envelope multiplied by an operator
required to be invariant under a rotation of the carrier phase, the weight-carrying
factor times an invariant, and the memory terms are required to vanish on a constant
envelope, so that the memoryless part is the CW characteristic. Third,
the choice of carrier is a choice of gauge, in the physicist's sense, a convention that
fixes a phase the physics leaves free: carrying the weight
on $\Atwo$ or on an external reference changes the coefficient functions and not the
object, exactly as moving a reference plane changes $S$-parameters. The factor $P^{h}$
is the gauge chosen in all three literatures: it is the phase normalization $P^{k}$ of
X-parameters~\cite{root2005phd,root2013}, the virtual harmonics generator of
HOSIDF~\cite{nuij2006}, and the phase vector of the Cardiff model.

\subsection{Counting the coefficients}

How many independent invariants and equivariants exist at each degree is answered by
the Molien series of the action (Appendix~\ref{app:molien}), and for two weight-one
variables the invariant ring has a simple structure: it is a free module of rank two over
the polynomial ring in three algebraically independent quadratic invariants, generated by
the constant $1$ and $\Iml(\Atwo\Ab_{1,1})$. Appendix~\ref{app:molien} records the series and the
primary/secondary (Hironaka) decomposition behind it; the appendix is used only for the
coefficient count cited in Section~\ref{sec:discussion}.

\section{The Exponents as Counters}
\label{sec:exponents}

\subsection{Cardiff terms as monomials}

Tian \emph{et al.} rewrite the Cardiff model as a polynomial in the load-side wave and its
conjugate, with the general term $\binom{m}{k}(\Atwo)^{m-k}(\Ab_{2,1})^{k}$, ``where the
value of $n$ can now be calculated with $n=m-2k$, the value of $m$ equals to the
polynomial order number''~\cite{tian2025}; Root \emph{et al.} had attributed the
conjugate terms of the X-parameter form to the non-analyticity of the linearization about
the large-signal operating point~\cite{root2005phd}. Write $A\equiv\Atwo$, $a=m-k$ and
$b=k$, so that a single term of \eqref{eq:cardiff} is the monomial
\begin{equation}
  |A|^{m}(\angle A)^{n}=A^{a}\Ab^{b},\qquad a=\tfrac{m+n}{2},\quad b=\tfrac{m-n}{2}.
  \label{eq:ab}
\end{equation}
Imposing the Cardiff relation $m=|n|+2r$ gives the identity
\begin{equation}
  r=\min(a,b).
  \label{eq:rmin}
\end{equation}
The index $r$ counts the conjugate pairs $A\Ab=|A|^{2}$ in the term: an $r=0$ term is a
pure power of $A$ or of $\Ab$, an $r=1$ term carries one factor $|A|^{2}$, and so on.
Table~\ref{tab:dictionary} lists the low-order terms in both notations, and
Fig.~\ref{fig:lattice} places the same terms on the $(a,b)$ lattice, where lines of constant $n$
run diagonally, lines of constant $m$ anti-diagonally, and $r$ is the distance from the
nearer axis. The restriction $r\le1$ retains only the first two rows and the first two
columns of the lattice, an L-shaped region along the axes, and excludes the whole
interior $a\ge2$, $b\ge2$.

\begin{table}[!t]
\centering
\caption{The monomials through degree three, and the first one excluded by $r\le1$, in
both notations ($b_{2,1}$-type responses).}
\label{tab:dictionary}
\begin{tabular}{llrrrl}
\toprule
monomial & Cardiff term & $m$ & $n$ & $r$ & interpretation \\
\midrule
$1$              & constant             & 0 & 0    & 0 & no load dependence \\
$A$              & $|A|(\angle A)$      & 1 & 1    & 0 & linear in load wave \\
$\Ab$            & $|A|(\angle A)^{-1}$ & 1 & $-1$ & 0 & linear, conjugate \\
$A^2$            & $|A|^2(\angle A)^2$  & 2 & 2    & 0 & second order \\
$A\Ab$           & $|A|^2$              & 2 & 0    & 1 & second order, one pair \\
$\Ab^2$          & $|A|^2(\angle A)^{-2}$ & 2 & $-2$ & 0 & second order, conjugate \\
$A^3$            & $|A|^3(\angle A)^3$  & 3 & 3    & 0 & third order \\
$A^2\Ab$         & $|A|^3(\angle A)$    & 3 & 1    & 1 & third order, one pair \\
$A\Ab^{2}$       & $|A|^3(\angle A)^{-1}$ & 3 & $-1$ & 1 & third order, one pair \\
$\Ab^3$          & $|A|^3(\angle A)^{-3}$ & 3 & $-3$ & 0 & third order, conjugate \\
$A^2\Ab^{2}$     & $|A|^4$              & 4 & 0    & 2 & excluded by $r\le1$ \\
\bottomrule
\end{tabular}
\end{table}

\begin{figure}[!t]
\centering
\begin{tikzpicture}[scale=0.56]
  \fill[black!9] (-0.45,-0.45) rectangle (5.45,1.45);
  \fill[black!9] (-0.45,-0.45) rectangle (1.45,5.45);
  \draw[black!18,step=1] (-0.45,-0.45) grid (5.45,5.45);
  \draw[-{Stealth[length=2.2mm]}] (-0.45,0)--(5.95,0) node[right]{$a$};
  \draw[-{Stealth[length=2.2mm]}] (0,-0.45)--(0,5.95) node[above]{$b$};
  \foreach \i in {0,...,5}{
    \node[below,font=\scriptsize] at (\i,-0.5) {\i};
    \node[left,font=\scriptsize] at (-0.5,\i) {\i};}
  \draw[black!65,dashed,thick] (4.25,-0.25)--(-0.25,4.25);
  \draw[black!65,dotted,thick] (1.85,-0.15)--(5.35,3.35);
  \foreach \a in {0,...,5}{\foreach \b in {0,...,5}{
     \pgfmathtruncatemacro{\rr}{min(\a,\b)}
     \ifnum\rr<2
        \fill[black] (\a,\b) circle (2.3pt);
     \else
        \fill[white] (\a,\b) circle (2.6pt);
        \draw[black!60] (\a,\b) circle (2.6pt);
     \fi}}
  \draw[very thick] (2,2) circle (6.5pt);
  \begin{scope}[font=\scriptsize]
    \fill[black] (6.6,5.15) circle (2.3pt);
    \node[anchor=west] at (6.85,5.15) {kept by $r\le1$};
    \fill[white] (6.6,4.5) circle (2.6pt);
    \draw[black!60] (6.6,4.5) circle (2.6pt);
    \node[anchor=west] at (6.85,4.5) {discarded};
    \draw[black!65,dashed,thick] (6.35,3.5)--(6.95,3.5);
    \node[anchor=west] at (7.1,3.5) {$m=a+b=4$};
    \draw[black!65,dotted,thick] (6.35,2.85)--(6.95,2.85);
    \node[anchor=west] at (7.1,2.85) {$n=a-b=2$};
    \node[anchor=west,align=left] at (6.5,1.6)
       {$A^{2}\bar{A}^{2}=|A|^{4}$\\ first term discarded};
    \draw[black!70] (6.45,1.6)--(2.35,2.0);
  \end{scope}
\end{tikzpicture}
\caption{Monomials $A^{a}\Ab^{b}$ on the exponent lattice. The shaded L-shaped region is
$\min(a,b)\le1$, the Cardiff restriction $r\le1$: it contains every monomial through
total degree 3, and $A^{2}\Ab^{2}=|A|^{4}$ is the first monomial it excludes.}
\label{fig:lattice}
\end{figure}
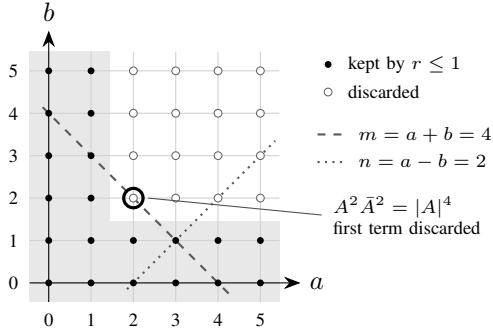

Through total degree $3$ the L-shaped region contains every monomial. Beyond that the
number of monomials retained by $r\le1$ grows linearly, $4D-2$, while the number of all
monomials of degree at most $D$ grows quadratically, $(D+1)(D+2)/2$: the restriction keeps $14$ of $15$ terms at $D=4$, $18$ of $21$ at
$D=5$, $34$ of $55$ at $D=9$ and $46$ of $91$ at $D=12$. The restriction therefore
removes many coefficients at high order and almost none at low order, which raises the
question of when the discarded terms are identically zero.

\subsection{The degree theorem}

Let the device be memoryless, with its dependence on the load-side voltage a polynomial
of degree $K$, and let the drive-side nonlinearity contribute harmonics of index $\ell$.
The memoryless hypothesis is used only here, because a polynomial degree in $v_2(t)$ is
a property of a static map. With
$v_2(t)=\Rel(Ae^{j\theta})$, $\theta=\omega_0 t$, the binomial theorem gives
\begin{equation}
  v_2^{\,p}=2^{-p}\sum_{q=0}^{p}\binom{p}{q}A^{q}\Ab^{\,p-q}e^{j(2q-p)\theta}:
  \label{eq:binom}
\end{equation}
a term of degree $p$ in the load-side voltage carries exactly $p$ factors drawn from
$\{A,\Ab\}$ and sits at harmonic offset $2q-p$. This is the frequency-mix bookkeeping of
weakly nonlinear analysis, in which an order-$p$ term reaches harmonic offset $k$ only if
$p\ge|k|$ and $p-k$ is even, with $(p+k)/2$ positive- and $(p-k)/2$ negative-frequency
factors~\cite{bussgang1974,lang1996}. Extracting the $h$-th harmonic requires
that $\ell+(2q-p)=h$, so with $a=q$, $b=p-q$,
\begin{equation}
  a+b=p,\qquad a-b=h-\ell .
  \label{eq:key}
\end{equation}

\begin{theorem}[Conjugate-pair bound, $r_{\max}=\lfloor K/2\rfloor$]\label{thm:rmax}
If a memoryless device's dependence on the load-side voltage is a polynomial of degree $K$, and
the load-side voltage is set by the incident wave, $v_2=\Rel(Ae^{j\theta})$, every term
in the response satisfies $\min(a,b)\le\lfloor K/2\rfloor$; in Cardiff
indices, $r_{\max}=\lfloor K/2\rfloor$.
\end{theorem}

\begin{proof}
By \eqref{eq:key}, $a+b=p\le K$, and $r=\min(a,b)\le(a+b)/2$. Since $r$ is an integer,
$r\le\lfloor K/2\rfloor$.
\end{proof}

The bound is elementary once $r$ is identified with $\min(a,b)$ in the rewriting
of~\cite{tian2025}; the content is that identification, which makes $r\le1$ a statement
about the degree of the load-side nonlinearity. We want to
emphasize four points about Theorem~\ref{thm:rmax}.
(i) The restriction $r\le1$ holds for every term if and only if $K\le3$: a cubic
load-side nonlinearity cannot supply the four factors that two conjugate pairs need, so
the restriction is exact for a cubic device and generically fails at fourth order, where
the first term to appear is $A^{2}\Ab^{2}=|A|^{4}$.
(ii) The bound is attained: a term with $\min(a,b)=r$ arises from $p=2r$ whenever the
drive side supplies $\ell=h$, which any nonlinearity with a linear term does at the
fundamental.
(iii) Separability is not assumed; the argument uses only the degree in $v_2$, so a
non-separable $F(v_1,v_2)$ with cross terms obeys the same bound. The theorem was checked
by brute-force expansion for $K=2,\dots,7$, separable and not.
(iv) Where the expansion is centered does not matter. A-pull data is naturally expanded
about the center $c$ of a circle rather than about $A=0$; substituting $A\mapsto A+c$
into $A^{\alpha}\Ab^{\beta}$ produces only terms with $a\le\alpha$, $b\le\beta$, so
$\min(a,b)$ cannot increase. One feature does change, and is needed in Section~\ref{sec:apull}: about $A=0$ a
source term of load-side order $p$ produces monomials with $a+b=p$ exactly, about
$A=c$ monomials with $a+b\le p$, so in a centered expansion the degree of a term is a
lower bound on the order that produced it, while the highest degree present is still
exactly $K$.

\subsection{Physical meaning of the three indices}

The identities \eqref{eq:key} give each of the three Cardiff indices a physical
meaning, stated as three corollaries.

\begin{corollary}\label{cor:m}
$m=|n|+2r=a+b=p$: the magnitude exponent is the order of the load-side nonlinearity that
produced the term, as Tian \emph{et al.}~\cite{tian2025} note for a polynomial in $\Atwo$ and its conjugate.
\end{corollary}
\begin{corollary}\label{cor:n}
$n=a-b=h-\ell$: the phase exponent is set by the harmonic the drive-side nonlinearity
supplied.
\end{corollary}
\begin{corollary}\label{cor:mix}
$m+|n-h|=p+|\ell|$ is the lowest total nonlinearity order able to produce the term,
load-side order plus drive-side harmonic. This is the ``mixing order'' by which the
Cardiff literature chooses its exponents~\cite{tian2025,alrawachy2025}; the identity is a
consistency check against their definition, and it holds on every row of Table~2
of~\cite{alrawachy2025}.
\end{corollary}

The practical difference is in
the reason for a truncation: the Cardiff literature limits $r$ to $0$ and $1$ on grounds of
extrapolation~\cite{tian2025}; Theorem~\ref{thm:rmax} makes the same limit a statement
about the device, that its load-side nonlinearity is cubic, which a measurement can check.

\section{Loaded Devices: Decay in $r$ and a Stopping Rule}
\label{sec:loaded}

Theorem~\ref{thm:rmax} concerns the \emph{intrinsic} nonlinearity. In a measurement the
load-side voltage is not independent: it satisfies $v_{ds}=a_2+b_2$, and $b_2$ on the
right-hand side is the device's own output, so the response is a fixed point rather
than a polynomial, and
a fixed point of even a quadratic map is a power series with every order present. A
loaded device has no finite $K$ and, strictly speaking, no finite $r_{\max}$.

In place of a bound, the coefficients decrease with $r$: every value of $r$ is present, but each
unit increase in $r$ reduces the largest coefficient, by $1.1$ to $1.6$ decades at
$\rho=0.5$ and $1.7$ to $2.2$ decades at $\rho=0.25$ over the first five steps of the example below, so the
series is truncated where the coefficients fall below the accuracy floor of the
measurement rather than at a fixed $r$. Consider a device whose intrinsic
dependence on the load-side voltage is quadratic, so that Theorem~\ref{thm:rmax} gives $r_{\max}=1$:
\begin{multline}
  i_d(v_1,v_2)=0.8v_1+0.2v_1^{2}-0.1v_1^{3}\\
  +0.08v_1v_2+0.04v_1^{2}v_2+0.05v_1v_2^{2}+0.04v_2^{2},
  \label{eq:toyid}
\end{multline}
in $Z_0=1$ units with the drive $v_1=\cos\theta$, $\theta=\omega_0t$, fixed. Inject only the fundamental at
port~2, $a_2(t)=\Rel(Ae^{j\theta})$, and close the loop pointwise in time,
$b_2=a_2-i_d(v_1,a_2+b_2)$.

For this device the series is known exactly. Write \eqref{eq:toyid} as
$i_d=c_0+c_1v_2+c_2v_2^{2}$, where $c_0$, $c_1$ and $c_2$ are polynomials in $v_1$ and
so functions of $\theta$ alone, and describe the device in a reference impedance $Z$,
with $v_2=\sqrt{Z}\,(a_2+b_2)$ and $i_d=(a_2-b_2)/\sqrt{Z}$; $Z=1$ is the case above. The
loop equation is then a quadratic in $v_2$ at each instant,
\begin{equation}
\begin{gathered}
  p\,v_2^{2}+q\,v_2+\sqrt{Z}\,c_0-2a_2=0,\\
  p=\sqrt{Z}\,c_2,\qquad q=\sqrt{Z}\,c_1+1/\sqrt{Z},
\end{gathered}
  \label{eq:loopquad}
\end{equation}
and the root that continues the small-signal solution is
\begin{equation}
  v_2=\frac{2\,(2a_2-\sqrt{Z}\,c_0)}{q+\sqrt{D_0+8p\,a_2}},\qquad D_0=q^{2}-4Zc_0c_2 .
  \label{eq:looproot}
\end{equation}
The binomial series of the square root in powers of $a_2$ gives each coefficient of
$B_{2,1}=\sum_{a,b}c_{a,b}A^{a}\Ab^{b}$ as an average over one period: for $k=a+b\ge1$,
\begin{equation}
\begin{aligned}
  c_{a,b}&=2^{1-k}\binom{k}{a}\,\bigl\langle g_k\,e^{j(a-b-1)\theta}\bigr\rangle_{\theta},\\
  g_k&=\frac{4}{\sqrt{Z}}\binom{1/2}{k}(8p)^{k-1}D_0^{\,1/2-k}-\delta_{k,1}.
\end{aligned}
  \label{eq:exactcoef}
\end{equation}
The series converges for $|A|<R$, where
\begin{equation}
  R(Z)=\min_{\theta}\frac{D_0}{8|p|};
  \label{eq:Rconv}
\end{equation}
at $|A|=R$ the discriminant $D_0+8p\,a_2$ reaches zero at one instant of the period for a
suitable phase of $A$, and for larger $|A|$ the loop equation has no real solution at that
instant. The closed-loop coefficients below are computed from \eqref{eq:exactcoef}; the
Taylor coefficients of \eqref{eq:looproot} were checked against a symbolic expansion.
Each coefficient is quoted in units of $\rho$, so that it is the size of its term at the
edge of the disc $|A|\le\rho$, and the largest coefficient in each $r$-class is recorded
relative to the constant term. In every class of the cases tabulated the largest term is
the diagonal one, $A^{r}\Ab^{r}=|A|^{2r}$. Table~\ref{tab:decay} and
Fig.~\ref{fig:decay} give the result for $Z=1$.

\begin{table}[!t]
\centering
\caption{Largest coefficient in each $r$-class, relative to the constant term, for the
quadratic device open- and closed-loop, $Z=1$. The closed-loop rows are from the exact
series \eqref{eq:exactcoef}. The open-loop row is a least-squares fit, in which
$10^{-14}$ is zero to rounding.}
\label{tab:decay}
\setlength{\tabcolsep}{3pt}
\begin{tabular}{lcccccc}
\toprule
$r$ & 0 & 1 & 2 & 3 & 4 & 5 \\
\midrule
open, $\rho=0.5$ & 1 & $8.6\mathrm{e}{-3}$ & $10^{-14}$ & $10^{-14}$ & $10^{-14}$ & $10^{-14}$ \\
closed, $\rho=0.5$ & 1 & $3.8\mathrm{e}{-2}$ & $1.1\mathrm{e}{-3}$ & $5.8\mathrm{e}{-5}$ & $4.3\mathrm{e}{-6}$ & $3.7\mathrm{e}{-7}$ \\
closed, $\rho=0.25$ & 1 & $9.4\mathrm{e}{-3}$ & $6.6\mathrm{e}{-5}$ & $9.1\mathrm{e}{-7}$ & $1.7\mathrm{e}{-8}$ & $3.6\mathrm{e}{-10}$ \\
\bottomrule
\end{tabular}
\end{table}

\begin{figure}[!t]
\centering
\begin{tikzpicture}
\begin{semilogyaxis}[
  width=\columnwidth, height=5.2cm,
  xlabel={conjugate-pair index $r$},
  ylabel={largest coefficient (relative)},
  xmin=-0.3, xmax=4.3, ymin=3e-7, ymax=4,
  xtick={0,1,2,3,4},
  grid=both, grid style={black!12},
  tick label style={font=\scriptsize}, label style={font=\footnotesize},
  legend columns=2,
  legend style={font=\scriptsize, draw=black!30, at={(0.5,1.04)}, anchor=south,
                /tikz/every even column/.append style={column sep=6pt}}]
\addplot[black, thick, mark=*, mark size=1.8pt]
  coordinates {(0,1) (1,3.8e-2) (2,1.1e-3) (3,5.8e-5) (4,4.3e-6)};
\addlegendentry{closed loop, $\rho=0.5$}
\addplot[black!45, thick, mark=o, mark size=1.8pt]
  coordinates {(0,1) (1,9.4e-3) (2,6.6e-5) (3,9.1e-7)};
\addlegendentry{closed loop, $\rho=0.25$}
\addplot[black, only marks, mark=square, mark size=3.2pt, thick]
  coordinates {(0,1) (1,8.6e-3)};
\addlegendentry{open loop, $\rho=0.5$}
\addplot[black!55, dashed, thick, domain=-0.3:4.3, samples=2] {1e-4};
\addlegendentry{illustrative floor, $10^{-4}$}
\end{semilogyaxis}
\end{tikzpicture}
\caption{The stopping rule for a loaded device: plot the largest coefficient at each $r$
on a log axis and truncate where the decay meets the floor. The open-loop device has no
terms beyond $r=1$. The floor drawn at $10^{-4}$ is illustrative.}
\label{fig:decay}
\end{figure}
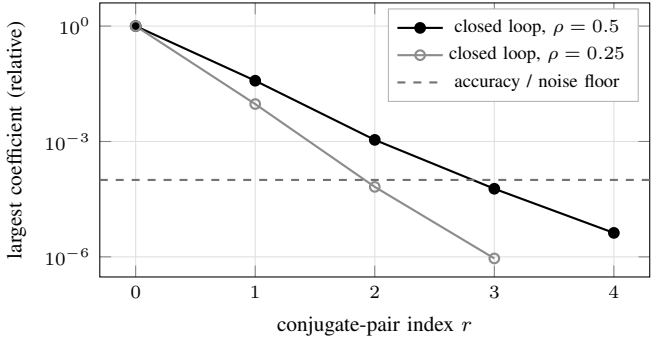

The open-loop row illustrates Theorem~\ref{thm:rmax} ($K=2$; nothing beyond $r=1$ to machine precision).
In the closed-loop rows each step in $r$ costs a factor $\rho^{2}$, from the extra pair
$|A|^{2}$, times the ratio of successive diagonal coefficients $|c_{r+1,r+1}/c_{r,r}|$,
which does not depend on $\rho$: $0.15$, $0.11$, $0.22$, $0.29$ and $0.34$ for the first
five steps. The ratio is not constant, so the decay is not geometric over this range.
Halving $\rho$ divides each step by four through the factor $\rho^{2}$ alone.

Table~\ref{tab:mismatch} gives the same ratios for five values of the reference impedance
$Z$, which shows the dependence on the load. The first step, which is present without the
loop, grows by a factor of about three between $Z=0.5$ and $Z=1.25$, a little faster than
in proportion to $Z$. The later steps, which exist only because of the loop, grow by a
factor of more than twenty. Multiplied by $R^{2}$, the ratios of the third, fourth and
fifth steps agree across the five impedances to within $7$, $4$ and $3$ percent. For this
device the dependence of the decay on the load is therefore carried by the radius of
convergence $R(Z)$, and the slope of the decay is set by $\rho/R$. Computed to $r=60$, the
diagonal ratios multiplied by $R^{2}$ rise steadily, to $0.77$ at the tenth step and
$0.96$ at the sixtieth, and from the tenth step on they agree across the five impedances
to within $0.005$. At $Z=1.5$ the radius of convergence is $0.50$ and the disc
$|A|\le0.5$ reaches it.

\begin{table}[!t]
\centering
\caption{Radius of convergence $R$ from \eqref{eq:Rconv} and ratios of successive largest
coefficients, divided by $\rho^{2}$, for the closed-loop quadratic device described in
reference impedance $Z$. The ratios are those of the diagonal coefficients and do not
depend on $\rho$.}
\label{tab:mismatch}
\setlength{\tabcolsep}{4pt}
\begin{tabular}{lcccccc}
\toprule
$Z$ & $R$ & $r{:}\,0\to1$ & $1\to2$ & $2\to3$ & $3\to4$ & $4\to5$ \\
\midrule
0.5  & 4.10 & 0.066 & 0.012 & 0.023 & 0.030 & 0.035 \\
0.75 & 2.15 & 0.104 & 0.041 & 0.081 & 0.107 & 0.125 \\
1    & 1.29 & 0.150 & 0.112 & 0.223 & 0.295 & 0.345 \\
1.25 & 0.81 & 0.215 & 0.274 & 0.554 & 0.740 & 0.869 \\
1.5  & 0.50 & 0.320 & 0.680 & 1.42  & 1.91  & 2.26  \\
\bottomrule
\end{tabular}
\end{table}

Table~\ref{tab:trunc} gives the error of the truncated series. The exact response,
from \eqref{eq:looproot} at $256$ instants per period, is compared with the series
restricted to $r\le r_{\max}$ on $1536$ points that cover the disc $|A|\le\rho$ with equal
area weight, and the entry is the normalized mean-square error. With every term to degree
$30$ the error at $Z=1$ is below $-290$\,dB, which checks the coefficients. At $Z=1$ each
added value of $r$ lowers the error by between $22$ and $45$\,dB, and the truncation
$r\le1$ leaves $-67$\,dB at $\rho=0.5$. At $Z=1.5$ and $\rho=0.5$, where the disc reaches
the radius of convergence, the three values $r=2$, $3$ and $4$ together lower the error by
$16$\,dB.

\begin{table}[!t]
\centering
\caption{Normalized mean-square error, in dB, of the exact series of the closed-loop
quadratic device truncated at $r\le r_{\max}$ (terms to degree $30$), over the disc
$|A|\le\rho$. At $Z=1.5$, $\rho=0.5$ the terms above degree $30$ contribute $-73$\,dB.}
\label{tab:trunc}
\setlength{\tabcolsep}{4pt}
\begin{tabular}{llccccc}
\toprule
$Z$ & $\rho$ & $r_{\max}=0$ & 1 & 2 & 3 & 4 \\
\midrule
0.5  & 0.5  & $-42.1$ & $-95.1$ & $-141.4$ & $-185.1$ & $-227.2$ \\
1    & 0.25 & $-45.6$ & $-90.8$ & $-129.3$ & $-165.1$ & $-199.2$ \\
1    & 0.5  & $-34.1$ & $-67.0$ & $-93.3$  & $-116.8$ & $-138.8$ \\
1.25 & 0.5  & $-30.5$ & $-54.9$ & $-72.7$  & $-87.9$  & $-101.5$ \\
1.5  & 0.5  & $-25.5$ & $-39.2$ & $-46.4$  & $-51.3$  & $-55.0$ \\
\bottomrule
\end{tabular}
\end{table}

A bound of the same form holds for any device whose response is analytic. Suppose
$B_{2,1}$, as a function of $A$ and $\Ab$ taken as independent complex variables, is
analytic and bounded by $M$ for $|A|<R$ and $|\Ab|<R$. Cauchy's estimate in two variables~\cite[Theorem~2.2.7]{hormander1990} gives
$|c_{a,b}|\le MR^{-(a+b)}$, and on the disc $|A|\le\rho<R$ the terms with $r>r_{\max}$,
which are those with $a>r_{\max}$ and $b>r_{\max}$, sum to at most
\begin{equation}
  \frac{M\,(\rho/R)^{2(r_{\max}+1)}}{(1-\rho/R)^{2}} .
  \label{eq:cauchy}
\end{equation}
The truncation error in $r$ is thus bounded by a geometric sequence of ratio
$(\rho/R)^{2}$. This is the standard estimate for a convergent power series. It is given
here for its form, since $M$ and $R$ are not known in advance for a measured device and
the estimate does not apply where the response is not analytic. In the example the
hypothesis holds with the $R(Z)$ of Table~\ref{tab:mismatch}, and the computed ratios stay
below $(\rho/R)^{2}$ and rise toward it.

For a loaded device $r_{\max}$ is set by where the decay crosses the accuracy floor: fit,
plot coefficient magnitude against $r$ on a log axis, and truncate where it meets the
floor. Theorem~\ref{thm:rmax} adds the intrinsic behavior: it states what the device would
produce without the loop.

\section{Tailored A-pull Measures the Decomposition}
\label{sec:apull}

\subsection{Clusters}\label{sec:clusters}

In a conventional load-pull sweep every term of \eqref{eq:cardiff} contributes to every
point, the coefficients are recovered by least squares, and the answer depends on the
basis assumed. The tailored A-pull method~\cite{alrawachy2025,tasker2020} steps the
load-side wave around a spiral about a chosen center $a^{c}_{2,1}$,
\begin{equation}
  a_{2,1}(I)=a^{c}_{2,1}+\tfrac12\Delta a\bigl[1+\cos(2\pi S_{am}I/N)\bigr]
  e^{j2\pi S_{pm}I/N},
  \label{eq:spiral}
\end{equation}
$I=0,\dots,N-1$, and Fourier transforms the response over the sequence index. The
model of~\cite{alrawachy2025}, its Eq.~(7), is written in the deviation
$A=a_{2,1}-a^{c}_{2,1}$, the centered expansion of Section~\ref{sec:exponents}. A term
$A^{a}\Ab^{b}=|A|^{m}e^{jn\phi}$ followed around the spiral contributes $e^{jn\phi}$
with $\phi$ advancing at rate $S_{pm}$ --- one line at $nS_{pm}$ --- times $|A|^{m}$,
a raised cosine at rate $S_{am}$ raised to the $m$-th power, which by the binomial expansion
\eqref{eq:binom} has harmonics up to order $m$ and no further. The result is a
\emph{cluster} of lines at
\begin{equation}
  sS_{am}+nS_{pm},\qquad s=-m,\dots,m,
  \label{eq:cluster}
\end{equation}
centered at $nS_{pm}$ with half-width $mS_{am}$ (Eq.~(8) of~\cite{alrawachy2025}). The center of
the cluster gives $n$ and its half-width gives $m$: the spectrum is a picture of the
exponent decomposition. The two rates $S_{am}$ and $S_{pm}$ are chosen so that clusters of neighboring $n$ do
not overlap: Al-Rawachy \emph{et al.}~\cite{alrawachy2025} require $S_{pm}>(2\,\mathrm{Order}-1)\,S_{am}$, their
Eq.~(9), where Order is the largest mixing order sought, and take $S_{am}$ and $S_{pm}$
prime so that no line of one cluster falls on a line of another.

\begin{proposition}[The cluster width is a device parameter]
By Corollary~\ref{cor:m} and point (iv) after Theorem~\ref{thm:rmax}, the half-width of
a cluster is the lowest load-side order able to produce it, and the widest cluster
present has half-width equal to the order $K$ of the load-side nonlinearity.
\end{proposition}

The half-width of the widest resolvable cluster thus gives $K$ by inspection of the
spectrum rather than by a fit, and Theorem~\ref{thm:rmax}
then gives $r_{\max}$.

\subsection{The published data}

Al-Rawachy \emph{et al.}~\cite{alrawachy2025} applied the tailored A-pull identification
of Section~\ref{sec:clusters} to a GaN device at $2.45$\,GHz and tabulated the extracted
coefficients of $b_{2,1}$. The first four
columns of Table~\ref{tab:alrawachy} are from~\cite{alrawachy2025}; the last two are computed here.

\begin{table}[!t]
\centering
\caption{The nine exponent pairs identified for $b_{2,1}$ on a GaN device at 2.45\,GHz.
$m$, $n$, mixing order and coefficient magnitude from Table~2 of~\cite{alrawachy2025} ($S_{am}=5$,
$S_{pm}=89$, $N=1069$); $r=(m-|n|)/2$ and the mixing order of Corollary~\ref{cor:mix}
($h=1$) computed here.}
\label{tab:alrawachy}
\begin{tabular}{rrcrrc}
\toprule
$m$ & $n$ & mix.\ order~\cite{alrawachy2025} & $|C_{2,1,m,n}|$ & $r$ & $m+|n-h|$ \\
\midrule
0 & $0$  & 1 & 1.379 & 0 & 1 \\
1 & $1$  & 1 & 0.462 & 0 & 1 \\
1 & $-1$ & 3 & 0.075 & 0 & 3 \\
2 & $0$  & 3 & 0.050 & 1 & 3 \\
2 & $2$  & 3 & 0.039 & 0 & 3 \\
2 & $-2$ & 5 & 0.007 & 0 & 5 \\
3 & $1$  & 3 & 0.027 & 1 & 3 \\
3 & $-1$ & 5 & 0.015 & 1 & 5 \\
3 & $3$  & 5 & 0.008 & 0 & 5 \\
\bottomrule
\end{tabular}
\end{table}

Every pair has $r\le1$ and the largest magnitude exponent is $m=3$; no pair with $m\ge4$
appears, in particular neither $(4,0)$ nor $(5,1)$, the two $r=2$ pairs that Al-Rawachy \emph{et al.}~\cite{alrawachy2025}
report searching for. By Corollary~\ref{cor:m} the load-side order detected above the
floor of that measurement is $3$. The table does not test Theorem~\ref{thm:rmax}: the
degree $K$ is read from the largest $m$ in the same table, and every table whose largest
$m$ is $3$ has $r\le1$. A test requires $K$ to be known independently, as it is in the
simulations of Section~\ref{sec:sims}, or a measurement in which the floor is lowered
until the first term with $m\ge4$ appears. The last column, computed from
Corollary~\ref{cor:mix}, reproduces the tabulated mixing order in all nine rows, which
checks that the two definitions of mixing order agree.

The same table gives a rough indication of where the first $r=2$ term may lie. The
coefficients of~\cite{alrawachy2025} multiply $|A|^{m}$ with $A$ in the wave units of that
paper, and the levels quoted here take $|A|=1$ in those units; at another amplitude a term
of order $m$ scales as $|A|^{m}$. At fixed $n$ the ratio of the $r=1$ to the $r=0$
coefficient is $0.050/1.379=0.036$ for $n=0$, $0.027/0.462=0.058$ for $n=1$ and
$0.015/0.075=0.20$ for $n=-1$. If the step from $(2,0)$ to $(4,0)$ had one of these
ratios, $(4,0)$ would lie at $-58$, $-54$ or $-43$\,dBc relative to $(0,0)$. In the
simulation of Section~\ref{sec:asmhemt} a one-step extrapolation of this kind is low by
$6$ to $16$\,dB, on lines that are themselves close to the numerical floor. The
identification floor of~\cite{alrawachy2025} is a signal-to-spurious ratio near
$-40$\,dBc, with a noise floor near $-60$\,dBc. The estimate is therefore a range, from
about $-58$\,dBc up to the spurious floor, and the non-detection of $(4,0)$ is consistent
with every level in that range. Al-Rawachy \emph{et al.} draw the qualitative conclusion
themselves, that ``the measured dynamic range is insufficient for all the 5th order model
terms to be accurately extracted''~\cite{alrawachy2025}.

If the spurious floor can be pushed toward the noise floor, the $(4,0)$ cluster should
be the first new term to appear, with $(5,1)$ the next $r=2$ term behind it. The level is
uncertain, as the range above shows, but the ordering is structural, since any $r=2$ term
needs load-side order at least four and $(4,0)$ is the one that needs no more. If nothing
appears down to the noise floor, $r\le1$ is confirmed for that device at that floor by
measurement rather than imposed. Either outcome is informative, and the measurement can be
made with the apparatus described in Section~III of~\cite{alrawachy2025}.

\section{Two Simulations}
\label{sec:sims}

\subsection{A polynomial device}
\label{sec:toy}

The results can be illustrated on a synthetic device. Consider a Curtice-type transistor
in $Z_0=1$ units, $i_d=g\,G(v_1)F_K(v_2)$ with $G(v_1)=\tanh(1.5(0.2+v_1))$ on the
drive side and $F_K$ the degree-$K$ Taylor polynomial of $\tanh(1+v_2)$ on the load
side, so that $K$ is the intrinsic load-side order; close it through the load as in
Section~\ref{sec:loaded}; apply \eqref{eq:spiral} with the published numbers
$S_{am}=5$, $S_{pm}=89$, six cycles, $N=1069$ points and load-wave amplitude up to
half the drive; Fourier transform $b_{2,1}$ over the sequence; and call a pair $(m,n)$
detected when its outermost line $nS_{pm}\pm mS_{am}$ stands $3$\,dB above a chosen
floor. Fig.~\ref{fig:spectra}(a) is the spectrum for $K=3$, which can be compared with
Fig.~6 of~\cite{alrawachy2025}. Sweeping the floor gives
Table~\ref{tab:floorsweep}.

\begin{table}[!t]
\centering
\caption{Floor sweep on the cubic polynomial device, closed loop.}
\label{tab:floorsweep}
\begin{tabular}{lp{0.62\columnwidth}}
\toprule
floor & newly detected pairs \\
\midrule
$-40$\,dBc & $(0,0)$, $(1,1)$, $(2,0)$ \\
$-50$\,dBc & $+\,(1,-1)$, $(2,2)$ \\
$-60$\,dBc & $+\,(3,1)$, $(3,-1)$: seven pairs, all $r\le1$, max $m=3$, hence
             $K_{\rm eff}=3$, $r_{\max}=1$ \\
$-70$\,dBc & $+\,(3,3)$, $(4,0)$ at $-59.6$\,dBc, $(4,2)$ at $-61.1$\,dBc, $(5,1)$ at
             $-66.1$\,dBc, $(2,-2)$ \\
\bottomrule
\end{tabular}
\end{table}

At a $-60$\,dBc floor the synthetic device gives seven of the
nine pairs of Table~\ref{tab:alrawachy}, every one with $r\le1$, largest magnitude
exponent three, and the inference $K_{\rm eff}=3$, $r_{\max}=1$. Ten decibels lower the
fourth-order content appears in the order argued in Section~\ref{sec:apull}, $(4,0)$ leading, then the $r=1$
term $(4,2)$, then $(5,1)$. The polynomial device puts $(4,0)$ at $-59.6$\,dBc. That level
moves with the load-wave amplitude and with the load, and no significance attaches to its
closeness to the lower end of the range estimated in Section~\ref{sec:apull}. The
structure does not move: the same run with the loop opened gives $(4,0)$ at $-300$\,dBc for
$K=3$ (a numerical zero, as Theorem~\ref{thm:rmax} requires) and at $-109$\,dBc for $K=4$. In this device
essentially all fourth-order content at the fundamental is generated through the load.

\begin{figure*}[!t]
\centering
\begin{tikzpicture}
\begin{axis}[
  width=0.5\textwidth, height=5.6cm,
  title={(a) polynomial device, $K=3$, closed loop},
  title style={font=\footnotesize},
  xlabel={bin index $n\,S_{pm}+s\,S_{am}$},
  ylabel={level re $(0,0)$, dBc},
  xmin=-230, xmax=230, ymin=0, ymax=96,
  ytick={0,20,40,60,80}, yticklabels={$-90$,$-70$,$-50$,$-30$,$-10$},
  xtick={-178,-89,0,89,178}, xticklabels={$-2S_{pm}$,$-S_{pm}$,$0$,$S_{pm}$,$2S_{pm}$},
  grid=major, grid style={black!12},
  tick label style={font=\scriptsize}, label style={font=\footnotesize},
  clip=true]
\addplot[ycomb, black, line width=0.5pt, mark=none] coordinates {
(-227,1.3) (-223,1.8) (-218,5.0) (-213,8.5) (-208,12.5) (-203,16.9) (-198,21.8) (-193,27.1) (-188,22.8) (-183,33.2) (-178,39.0) (-173,33.2) (-168,22.8) (-163,27.1) (-158,21.8) (-153,16.8) (-148,12.4) (-143,8.2) (-139,0.2) (-138,4.5) (-134,3.4) (-133,1.1) (-129,7.0) (-124,11.1) (-119,15.8) (-114,21.2) (-109,27.4) (-104,35.9) (-99,44.8) (-94,49.3) (-89,50.4) (-84,49.3) (-79,44.8) (-74,35.9) (-69,27.4) (-64,21.1) (-59,15.7) (-54,11.0) (-50,0.8) (-49,6.8) (-45,4.3) (-44,3.0) (-40,8.2) (-35,12.6) (-30,17.6) (-25,23.4) (-20,30.4) (-15,38.2) (-10,53.5) (-5,63.4) (0,90.0) (5,63.4) (10,53.5) (15,38.2) (20,30.4) (25,23.4) (30,17.6) (35,12.5) (39,0.9) (40,8.1) (44,4.4) (45,4.2) (49,8.4) (50,0.6) (54,12.9) (59,17.9) (64,23.9) (69,30.8) (74,40.7) (79,50.6) (84,78.4) (89,84.2) (94,78.4) (99,50.6) (104,40.7) (109,30.8) (114,23.9) (119,18.0) (124,12.9) (128,0.3) (129,8.4) (133,3.8) (134,4.5) (138,7.6) (139,0.9) (143,11.9) (148,16.8) (153,22.3) (158,28.9) (163,36.2) (168,49.0) (173,58.1) (178,61.1) (183,58.1) (188,49.0) (193,36.2) (198,28.9) (203,22.3) (208,16.8) (213,12.0) (218,7.7) (222,2.3) (223,3.9) (227,5.9) (228,0.5)
};
\addplot[black!60, dashed, thick] coordinates {(-230,50) (230,50)};
\addplot[black!60, dotted, thick] coordinates {(-230,30) (230,30)};
\node[font=\tiny, anchor=south west, fill=white, inner sep=1pt] at (axis cs:-228,50) {spurious floor $-40$\,dBc};
\node[font=\tiny, anchor=south west, fill=white, inner sep=1pt] at (axis cs:-228,30) {noise floor $-60$\,dBc};
\node[font=\tiny, anchor=south, fill=white, inner sep=1pt] at (axis cs:14,88)   {$(0,0)$};
\node[font=\tiny, anchor=south, fill=white, inner sep=1pt] at (axis cs:89,79)   {$(1,1)$};
\node[font=\tiny, anchor=south, fill=white, inner sep=1pt] at (axis cs:-89,57)  {$(1,-1)$};
\node[font=\tiny, anchor=south, fill=white, inner sep=1pt] at (axis cs:178,57)  {$(2,2)$};
\node[font=\tiny, anchor=south, fill=white, inner sep=1pt] at (axis cs:30,55)   {$(2,0)$};
\node[font=\tiny, anchor=south, fill=white, inner sep=1pt] at (axis cs:128,45)  {$(3,1)$};
\node[font=\tiny, anchor=south, fill=white, inner sep=1pt] at (axis cs:-60,37)  {$(3,-1)$};
\node[font=\tiny, anchor=south, fill=white, inner sep=1pt] at (axis cs:-40,26)  {$(4,0)$};
\end{axis}
\end{tikzpicture}%
\hfill
\begin{tikzpicture}
\begin{axis}[
  width=0.5\textwidth, height=5.6cm,
  title={(b) ASM-HEMT GaN model in ngspice, 0.8\,V drive},
  title style={font=\footnotesize},
  xlabel={bin index $n\,S_{pm}+s\,S_{am}$},
  ylabel={level re $(0,0)$, dBc},
  xmin=-230, xmax=230, ymin=0, ymax=106,
  ytick={0,20,40,60,80,100}, yticklabels={$-100$,$-80$,$-60$,$-40$,$-20$,$0$},
  xtick={-178,-89,0,89,178}, xticklabels={$-2S_{pm}$,$-S_{pm}$,$0$,$S_{pm}$,$2S_{pm}$},
  grid=major, grid style={black!12},
  tick label style={font=\scriptsize}, label style={font=\footnotesize},
  clip=true]
\addplot[ycomb, black, line width=0.5pt, mark=none] coordinates {
(-227,0.9) (-218,2.1) (-208,13.0) (-203,5.1) (-198,25.1) (-194,6.9) (-189,13.4) (-188,24.6) (-184,16.3) (-183,48.8) (-179,17.2) (-178,53.8) (-174,16.3) (-173,48.8) (-169,13.4) (-168,24.6) (-164,6.8) (-158,25.1) (-153,5.1) (-148,12.9) (-138,2.2) (-129,2.9) (-119,14.0) (-114,7.9) (-109,28.8) (-104,17.4) (-100,8.2) (-99,52.2) (-95,12.7) (-94,68.0) (-93,5.5) (-90,14.1) (-89,75.3) (-88,7.2) (-85,12.8) (-84,68.0) (-83,5.5) (-80,8.2) (-79,52.2) (-74,17.4) (-69,28.8) (-64,7.9) (-59,14.0) (-49,2.9) (-40,3.1) (-30,14.3) (-25,9.0) (-20,28.2) (-15,34.6) (-10,55.9) (-9,8.2) (-6,1.6) (-5,68.4) (-4,12.2) (-1,4.2) (0,100.0) (1,13.4) (4,1.7) (5,68.4) (6,12.2) (10,55.9) (11,8.2) (15,34.6) (20,28.2) (25,9.0) (30,14.3) (40,3.2) (49,2.8) (59,13.7) (64,11.6) (69,27.7) (74,31.8) (75,4.6) (78,1.1) (79,41.5) (80,10.5) (83,0.5) (84,77.4) (85,13.1) (89,83.3) (90,13.8) (93,0.4) (94,77.4) (95,13.1) (98,1.0) (99,41.5) (100,10.5) (103,0.1) (104,31.8) (105,4.6) (109,27.7) (114,11.5) (119,13.7) (129,2.8) (138,2.0) (148,12.2) (153,12.4) (158,24.5) (163,33.3) (164,5.0) (167,2.7) (168,49.2) (169,5.2) (172,4.7) (173,62.4) (177,5.3) (178,66.0) (182,4.7) (183,62.4) (187,2.6) (188,49.2) (189,5.1) (193,33.3) (194,5.0) (198,24.5) (203,12.4) (208,12.2) (218,2.0) (227,0.3)
};
\addplot[black!60, dashed, thick] coordinates {(-230,60) (230,60)};
\addplot[black!60, dotted, thick] coordinates {(-230,40) (230,40)};
\node[font=\tiny, anchor=south west, fill=white, inner sep=1pt] at (axis cs:-228,60) {spurious floor $-40$\,dBc};
\node[font=\tiny, anchor=south west, fill=white, inner sep=1pt] at (axis cs:-228,40) {noise floor $-60$\,dBc};
\node[font=\tiny, anchor=south, fill=white, inner sep=1pt] at (axis cs:14,99)   {$(0,0)$};
\node[font=\tiny, anchor=south, fill=white, inner sep=1pt] at (axis cs:89,80)   {$(1,1)$};
\node[font=\tiny, anchor=south, fill=white, inner sep=1pt] at (axis cs:-89,71)  {$(1,-1)$};
\node[font=\tiny, anchor=south, fill=white, inner sep=1pt] at (axis cs:30,58)   {$(2,0)$};
\node[font=\tiny, anchor=south, fill=white, inner sep=1pt] at (axis cs:178,52)  {$(2,2)$};
\node[font=\tiny, anchor=south, fill=white, inner sep=1pt] at (axis cs:125,35)  {$(3,1)$};
\node[font=\tiny, anchor=south, fill=white, inner sep=1pt] at (axis cs:-40,31)  {$(4,0)$};
\end{axis}
\end{tikzpicture}
\caption{Simulated tailored A-pull spectra of $b_{2,1}$, $|k|\le230$ shown, with the two
floors of~\cite{alrawachy2025} drawn in. Each cluster is centered at $nS_{pm}$ and has
half-width $mS_{am}$, $S_{am}=5$, $S_{pm}=89$. (a) The cubic polynomial device of
Section~\ref{sec:toy}: the $(2,0)$ cluster spans $\pm10$ bins about zero, the $(3,1)$
cluster $\pm15$ about 89; the $(4,0)$ line at $\pm20$ bins, $-59.6$\,dBc, is the first
fourth-order content. (b) The ASM-HEMT model of Section~\ref{sec:asmhemt} at the softer
drive: bins belonging to no cluster with $|n|\le8$ are below $-80$\,dBc.}
\label{fig:spectra}
\end{figure*}
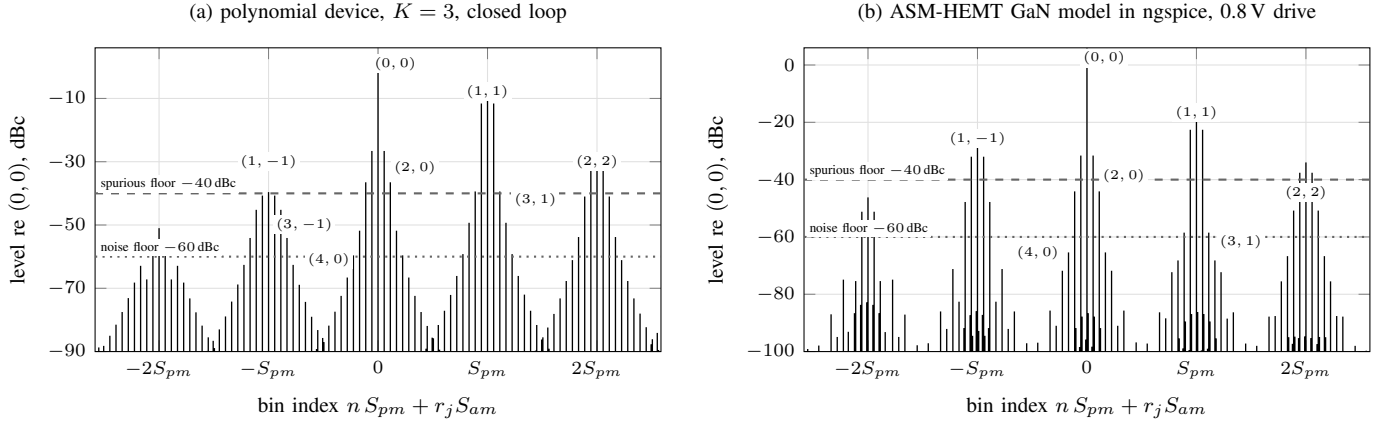

\subsection{A compact transistor model in a circuit simulator}
\label{sec:asmhemt}

The polynomial device has the right structure but is not a realistic transistor. The same
experiment was therefore run on a physics-based compact model in an open-source
simulator: the ASM-HEMT GaN HEMT model (version 101.4.0 of the model described
in~\cite{khandelwal2018asmhemt};
Verilog-A), compiled with OpenVAF and run in ngspice~44 through its OSDI interface. The
default parameter set describes a $0.8$\,mm device with $V_{\rm off}=-2$\,V; it was
biased at $V_{gs}=-1.5$\,V (class AB, about $170$\,mA) with about $20$\,V on the drain,
driven at $2.45$\,GHz from a $50\,\Omega$ source into a $50\,\Omega$ environment.
Self-heating and trapping were switched off so that the device is memoryless, as
Theorem~\ref{thm:rmax} assumes, and reaches its CW steady state within a few periods. A-pull is simple in simulation, as Al-Rawachy \emph{et al.} note: a Thevenin
source $2\sqrt{Z_0}\Rel(A_2e^{j\omega t})$ behind $Z_0$ at the drain sets the incident
wave $a_2$ exactly, whatever the device does. The maximum-power load was found by a scan,
the spiral \eqref{eq:spiral} wound about it with $S_{am}=5$, $S_{pm}=89$, six cycles and
$N=1069$ points, its radius enclosing roughly the $3$\,dB output-power contour, the
criterion of~\cite{alrawachy2025}. Each point is a short transient (eight periods, the last four
analyzed) at a fixed step of $1024$ per period.\footnote{The netlists, scripts and both datasets, together with the MATLAB
routines that reproduce the tables and figures of this paper, are archived
in~\cite{toolkit}. The closed-loop rows of Table~\ref{tab:decay} and
Tables~\ref{tab:mismatch} and~\ref{tab:trunc} are produced by \texttt{loaded\_exact.m},
added in version 1.2.0.} Fig.~\ref{fig:spectra}(b) is the spectrum at a drive of $0.8$\,V source voltage, about
$1$--$2$\,dB into compression, and Table~\ref{tab:asmhemt} collects the outermost-line
levels at two drives, arranged by $n$ so that the decay in $r$ can be read along each
row. The levels come from a single simulator and a single method, a transient run to
steady state at each point; they have not been cross-checked by harmonic balance or
against a measurement.

\begin{table}[!t]
\centering
\caption{ASM-HEMT: outermost-line level of each $(m,n)$ cluster, dBc relative to
$(0,0)$, at two drives. Along a row $m$ increases by two, i.e.\ $r$ by one. Entries in
parentheses lie below the largest line that sits on no cluster position with $|n|\le8$,
$-80$\,dBc at the softer drive and $-107$\,dBc at the harder, and are not to be read as
device response. Counting clusters to $|n|\le6$ only, as the archived MATLAB routine does,
that level is $-72$\,dBc and $-92$\,dBc.}
\label{tab:asmhemt}
\setlength{\tabcolsep}{3pt}
\begin{tabular}{lrrrcrrr}
\toprule
& \multicolumn{3}{c}{0.8\,V (1--2\,dB compr.)} & &
  \multicolumn{3}{c}{1.2\,V (3--4\,dB compr.)} \\
\cmidrule{2-4}\cmidrule{6-8}
$n$ & $m{=}|n|$ & $|n|{+}2$ & $|n|{+}4$ & & $m{=}|n|$ & $|n|{+}2$ & $|n|{+}4$ \\
\midrule
$0$  & $0.0$   & $-44.1$ & $-71.8$   & & $0.0$   & $-51.8$ & $-97.6$  \\
$+1$ & $-22.6$ & $-68.2$ & $(-88.4)$ & & $-27.2$ & $-81.1$ & $(-116)$ \\
$-1$ & $-32.0$ & $(-82.6)$ & $(-92.1)$ & & $-30.8$ & $-71.8$ & $(-113)$ \\
$+2$ & $-50.8$ & $-75.5$ & $(-87.8)$ & & $-57.2$ & $-100$  & $(-131)$ \\
$-2$ & $-75.4$ & $-74.9$ & $(-87.0)$ & & $-56.8$ & $-91.3$ & $(-131)$ \\
$+3$ & $-74.5$ & $(-86.9)$ & $(-101)$ & & $-80.3$ & $(-118)$ & $(-150)$ \\
$-3$ & $-73.7$ & $(-92.4)$ & $(-107)$ & & $-87.3$ & $(-110)$ & $(-141)$ \\
\bottomrule
\end{tabular}
\end{table}

Four points can be made about the ASM-HEMT spectra. (i) Counting clusters to $|n|\le6$,
the positions \eqref{eq:cluster} cover $311$ of the $1069$ bins, and the largest line off
that grid is at $-72$\,dBc at the softer drive and $-92$\,dBc at the harder. Those lines
sit on the positions of the $|n|=7$ and $8$ clusters, which fall one bin from the centers
of the $n=\mp5$ and $\mp4$ clusters because $12\,S_{pm}=N-1$; counting clusters to
$|n|\le8$ ($471$ bins) the largest remaining line is at $-80$\,dBc and $-107$\,dBc. Every
line above those levels sits on a cluster position. Since $S_{am}$ and $N$ are coprime the
positions fill every bin as the order is raised, so this check has force only at low
order. At the softer drive the center lines of the clusters with $n=4$ to $7$ lie between
$-64$ and $-72$\,dBc and do not fall with $n$, whereas at the harder drive they fall by
about $10$\,dB per unit of $n$. A rerun of the softer drive with twice the steps per
period, a tenfold smaller relative tolerance and twelve periods in place of eight changes
no spectral line by more than $-107$\,dB relative to $(0,0)$; the softer-drive levels
therefore come from the model as simulated. A likely cause is that the softer-drive
spiral straddles the onset of knee clipping. Over every third point of that spiral the
minimum drain voltage in a period ranges from $0.05$ to $9$\,V, and $14\%$ of the points
stay above $1$\,V, whereas at the harder drive it is below $0.75$\,V at every point
checked. This explanation has not been tested further.
(ii) At the published floors this device gives $K_{\rm eff}=2$, $r_{\max}=1$: with a
$-60$\,dBc floor and the $3$\,dB margin of Section~\ref{sec:toy} five pairs are detected
at each drive, and they are not the same five: $(0,0)$, $(1,\pm1)$, $(2,0)$ and $(2,2)$ at
the softer drive, and $(0,0)$, $(1,\pm1)$, $(2,0)$ and $(2,-2)$ at the harder, where
$(2,2)$ at $-57.2$\,dBc misses the margin by $0.2$\,dB. All have $r\le1$ and none exceeds
$m=2$. The default ASM-HEMT has a weaker
load-side nonlinearity than the device in~\cite{alrawachy2025}, whose $m=3$ terms lie near
$-35$ to $-45$\,dBc; here they are at $-68$\,dBc and below, so the simulation does not
reproduce the measured set of pairs. (iii) $(4,0)$ is the only
$r=2$ term above the floor at either drive: $-71.8$\,dBc at the softer drive, $8$\,dB above
the $|n|\le8$ level and equal to the $|n|\le6$ level, and $-97.6$\,dBc at the harder,
$9$\,dB above the $|n|\le8$ level and below the $|n|\le6$ level. The other $r=2$ lines are
below the floor, so the simulation is consistent with the ordering argued in
Section~\ref{sec:apull} and cannot confirm it. Among the lines with $m\ge4$ that are above
the floor, $(4,0)$ is the strongest at the softer drive, while at the harder drive the
$r=1$ term $(4,-2)$ exceeds it, the interleaving the polynomial device also showed. (iv)
The decay in $r$ is not geometric: along $n=0$ the first two steps are $-44$ and
$-28$\,dB at the softer drive and $-52$ and $-46$\,dB at the harder. A one-step
extrapolation from the first ratio, as done in Section~\ref{sec:apull} with the measured
coefficients, would put $(4,0)$ at $-88$\,dBc and $-104$\,dBc where the model has it at
$-72$\,dBc and $-98$\,dBc; the extrapolation is low by $16$\,dB and $6$\,dB, both figures
resting on lines within $10$\,dB of the floor.

A run-to-run jitter of a few milliradians in the phase reference of the phasor
extraction produces a flat floor of lines at $-55$ to $-70$\,dBc at every bin, on the
cluster positions and in the empty bins between them, under which every $m\ge3$ term is buried: the spurious floor of a
phase-referenced measurement. The level of the bins that belong to no cluster is
therefore the first quantity to check; those bins must be empty, and nothing below
their level should be read as device response. The DC drain current, read the same way,
has clusters at $n$ and $-n$ of identical level ($(1,\pm1)$ at $-29.0$\,dBc, $(2,\pm2)$ at $-62.2$\,dBc), the
statement that a weight-zero equivariant of a real quantity has conjugate-symmetric
coefficients.

\section{Discussion}\label{sec:discussion}

Three questions in the practice of behavioral modeling are answered by rules stated within
each modeling framework: the Cardiff literature limits $r$ to $0$ and $1$
``because of extrapolation concerns'' and chooses exponents by mixing
order~\cite{tian2025,alrawachy2025}, and the predistortion literature prunes its basis by
sparse regression~\cite{becerra2018}. The results above bear on two of them: how many coefficients
exist at a given order (the Molien series, Appendix~\ref{app:molien}, and the counts of
Section~\ref{sec:exponents}), and at what $r$ to truncate (Theorem~\ref{thm:rmax} for
an intrinsic nonlinearity, the decay rule of Section~\ref{sec:loaded} for a loaded
device). The third --- how many harmonics to keep --- is addressed by the roll-off theory
of linear time-periodic systems~\cite{sandberg2005}, since the linearization of
\eqref{eq:complete} about a drive is a harmonic transfer matrix~\cite{wereley1990}. In
that theory the entries of the harmonic transfer matrix decay in both of its directions.
The decay with harmonic index is set by the smoothness of the periodic time variation,
the decay with frequency by the leading Markov parameters, and the theory bounds the
error of the truncated matrix. The number of harmonics to keep is therefore a measurable
property of the driven device, like the decay in $r$ here. A fourth question, how many memory states a dynamic model
needs, is the subject of the companion paper~\cite{tufillaro2026memory}.

The scope of this paper is a two-port device in fundamental load-pull under single-tone
CW excitation, in periodic steady state, and, for Theorem~\ref{thm:rmax}, with a
memoryless load-side nonlinearity. The multi-harmonic form of the Cardiff model, one
phase exponent per harmonic, is written out in~\cite{tasker2011,bell2014}. A second
extension is two-tone excitation at incommensurate frequencies, where each tone carries
its own clock: the multi-tone multi-harmonic scattering (M$^2$S) parameters of Gibiino
\emph{et al.}~\cite{gibiino2020} index the response by the pair of intermodulation orders
and have coefficients that depend on the two magnitudes alone, ``due to time-invariance''
in the authors' words. The multi-harmonic extension uses the same circle: the injected
harmonics are commensurate with the drive and carry weights $2,3,\dots$, so
Theorem~\ref{thm:complete} applies with more weights and more indices to track. The
two-tone extension is different in kind: incommensurate tones carry independent phase
origins, and the argument is repeated with a two-torus in place of the circle, as in
Section~\ref{sec:symmetry}. Everything
quantitative here rests on simulation and on one published table, and the transistor
simulation comes from a single simulator without a harmonic-balance cross-check; the measurement that
would settle the interpretation of $r\le1$ --- a fundamental load-pull sweep at one bias and
drive with the spurious floor pushed toward the noise floor --- is within reach of the
apparatus described in Section~III of~\cite{alrawachy2025}.

\section{Conclusion}

Time invariance requires the response of a two-port under single-tone excitation to be a
phase factor on the fundamental times a function of the two wave magnitudes and their
relative phase, the form the Cardiff, X-parameter and describing-function models use.
Theorem~\ref{thm:complete} collects that statement and is not new. In the monomial
rewriting the three Cardiff exponents count the order of the load-side nonlinearity, the
drive-side harmonic and the conjugate pairs, and $r\le\lfloor K/2\rfloor$ for an intrinsic
load-side nonlinearity of degree $K$. These statements suggest the following procedure for
identifying a Cardiff model; it has not been tested on measured data. Convert Cardiff
exponents to $(a,b)$ and take $r=\min(a,b)$. Interpret $m$ as an order: the largest $m$
resolved is the load-side nonlinearity order at that drive and mismatch. Treat $r\le1$ as
a hypothesis about the device, exact for $K\le3$; if $r=2$ terms are needed to fit the
data, the device has fourth-order load-side content. For a loaded device plot the
coefficients against $r$ on a log axis and truncate where the line meets the floor,
reporting the slope. In the example of Section~\ref{sec:loaded} the step ratio tends to
$(\rho/R)^{2}$, with $R$ the radius of convergence of the series, so the decay slows as
the load-pull radius approaches $R$. With A-pull apparatus, measure the cluster widths, which give $K$ by
inspection of the spectrum rather than by a fit, and check the off-cluster bins first.
Finally, distinguish two kinds of residual. One that shrinks as terms are added is
truncation, which Section~\ref{sec:loaded} shows is always present for a loaded device.
One that stays at the same level as terms are added lies outside the model class: another
tone, another port, harmonic injection, or a measurement that is not in steady state.

\appendices

\section{Molien Series and the Hironaka Decomposition}
\label{app:molien}
For a circle acting on $\mathbb C^{N}$ with weights $w_1,\dots,w_N$ (and on the
conjugates with $-w_i$), the number of invariant polynomials of degree $d$ is the
coefficient of $t^{d}$ in
\begin{equation}
  M(t)=\bigl[\lambda^{0}\bigr]\prod_i\frac{1}{(1-t\lambda^{w_i})(1-t\lambda^{-w_i})},
\end{equation}
where $[\lambda^{0}]$ extracts the constant term in $\lambda$ (the Molien--Weyl
integral over the circle); replacing $\lambda^{0}$ by $\lambda^{h}$ counts the
weight-$h$ equivariants. For the two weight-one variables $\Aone,\Atwo$,
\begin{equation}
  M(t)=\frac{1+t^{2}}{(1-t^{2})^{3}} .
\end{equation}
The denominator says there are three algebraically independent invariants of degree
two, $|\Aone|^{2}$, $|\Atwo|^{2}$ and $\Rel(\Atwo\Ab_{1,1})$, in which arbitrary
polynomials may be formed; the numerator says there is one more, $\Iml(\Atwo\Ab_{1,1})$,
appearing at most linearly because its square is a polynomial in the other three
($(\Iml)^2=|\Aone|^2|\Atwo|^2-(\Rel)^2$). Every invariant is uniquely
$p+q\,\Iml(\Atwo\Ab_{1,1})$ with $p,q$ polynomials in the three primaries: the invariant
ring is a free module of rank two over the polynomial ring in the primaries (the Hironaka decomposition into
primary and secondary invariants~\cite{sturmfels1993,derksen2015}; a readable short
account, aimed at physicists, is~\cite{demellokoch2026}). The weight-$h$ equivariants
form a module over this ring generated, for $h=1$, by $\Aone$ and $\Atwo$, and for
general $h$ by the monomials $\Aone^{h-j}\Atwo^{j}$; the Cardiff basis
\eqref{eq:cardiff} is a generating set of this module written in the gauge $P$,
found experimentally one harmonic at a time. Numerically the completeness of
\eqref{eq:cardiff} at each degree is the statement that the design matrix of weight-$h$
monomials evaluated at random sample points has rank equal to the Molien count, which
the routines of~\cite{toolkit} verify.

\section*{Acknowledgment}

The author thanks the two anonymous reviewers of an earlier version of this paper, whose
reports identified errors that are corrected here. The author used Claude (Anthropic) to
edit for clarity, check algebra, and assist with the figures and simulation scripts; the
author reviewed and validated all AI-assisted output and takes full responsibility for
the content.


\end{document}